\documentclass[twoside,11pt]{article}

\usepackage{jmlr2e}

\usepackage{url}
\usepackage{algorithm}
\usepackage{algorithmicx}
\usepackage{algpseudocode}
\usepackage{amsmath}
\usepackage{tikz}
\usepackage{pgfplots}
\usepackage{multirow}
\usepackage{comment}
\usepackage{paralist}
\usepackage{hyperref}
\usepackage{etoolbox}

\algblockdefx[RepeatWhile]{RepeatWhile}{EndRepeatWhile}[0]{\textbf{repeat}}[1]{\textbf{while} #1}

\pgfplotsset{compat=newest}
\pgfplotscreateplotcyclelist{datasets}{%
  purple,every mark/.append style={fill=.!80!purple},mark=o\\%
  red,every mark/.append style={fill=.!80!red},mark=*\\%
  blue,every mark/.append style={fill=.!80!blue},mark=square*\\%
  green!60!black,every mark/.append style={fill=.!80!black},mark=triangle*\\%
  black,mark=star\\%
}
\pgfplotscreateplotcyclelist{datasets-methods}{%
  mark=square*,   blue!90!black, solid, every mark/.append style={solid}\\%
  mark=square,    blue!90!black, solid, every mark/.append style={solid}\\%
  mark=*,         red!80!black,  solid, every mark/.append style={solid}\\%
  mark=o,         red!80!black,  solid, every mark/.append style={solid}\\%
  mark=triangle*, green!50!black,  solid, every mark/.append style={solid}\\%
  mark=triangle,  green!50!black,  solid, every mark/.append style={solid}\\%
  mark=diamond*,  brown!80!black, solid, every mark/.append style={solid}\\%
  mark=diamond,   brown!80!black, solid, every mark/.append style={solid}\\%
  mark=pentagon*, purple!60!black, solid, every mark/.append style={solid}\\%
  mark=pentagon,  purple!60!black, solid, every mark/.append style={solid}\\%
  mark=star,      teal!60!black, solid, every mark/.append style={solid}\\%
}
\pgfplotscreateplotcyclelist{datasets-methods-5k}{%
  mark=square*,   blue!90!black, solid, every mark/.append style={solid}\\%
  mark=square,    blue!90!black, solid, every mark/.append style={solid}\\%
  mark=square*,   blue!90!black, dashed, every mark/.append style={solid}\\%
  mark=square,    blue!90!black, dashed, every mark/.append style={solid}\\%
  mark=*,         red!80!black,  solid, every mark/.append style={solid}\\%
  mark=o,         red!80!black,  solid, every mark/.append style={solid}\\%
  mark=*,         red!80!black,  dashed, every mark/.append style={solid}\\%
  mark=o,         red!80!black,  dashed, every mark/.append style={solid}\\%
  mark=triangle*, green!50!black,  solid, every mark/.append style={solid}\\%
  mark=triangle,  green!50!black,  solid, every mark/.append style={solid}\\%
  mark=triangle*, green!50!black,  dashed, every mark/.append style={solid}\\%
  mark=triangle,  green!50!black,  dashed, every mark/.append style={solid}\\%
  mark=diamond*,  brown!80!black, solid, every mark/.append style={solid}\\%
  mark=diamond,   brown!80!black, solid, every mark/.append style={solid}\\%
  mark=diamond*,  brown!80!black, dashed, every mark/.append style={solid}\\%
  mark=diamond,   brown!80!black, dashed, every mark/.append style={solid}\\%
  mark=pentagon*, purple!60!black, solid, every mark/.append style={solid}\\%
  mark=pentagon,  purple!60!black, solid, every mark/.append style={solid}\\%
  mark=pentagon*, purple!60!black, dashed, every mark/.append style={solid}\\%
  mark=pentagon,  purple!60!black, dashed, every mark/.append style={solid}\\%
  mark=star,      teal!60!black, solid, every mark/.append style={solid}\\%
}

\def\tblskip{\vspace\tblskipamount}
\newskip\tblskipamount
\newcommand{\hlinetop}{\hline\noalign{\tblskip}}
\newcommand{\hlinemid}{\noalign{\tblskip}\hline\noalign{\tblskip}}
\newcommand{\hlinebot}{\noalign{\tblskip}\hline}

\newcommand{\Tblref}[1]{Table~\ref{#1}}
\newcommand{\tblref}[1]{Table~\ref{#1}}
\newcommand{\figref}[1]{Figure~\ref{#1}}
\newcommand{\Figref}[1]{Figure~\ref{#1}}
\newcommand{\thmref}[1]{Theorem~\ref{#1}}
\newcommand{\secref}[1]{Section~\ref{#1}}
\renewcommand{\algref}[1]{Algorithm~\ref{#1}}

\newcommand{\inner}[2]{\langle#1,#2\rangle}
\newcommand{\tr}{{}^T}

\DeclareMathOperator*{\argmin}{argmin}
\newcommand{\score}{\phi}
\newcommand{\nodes}{V}
\newcommand{\set}[1]{\uppercase{#1}}
\renewcommand{\vec}[1]{\mathbf{#1}}
\newcommand{\setc}{\set{c}}
\newcommand{\setctrue}{\set{c}^*}

\newcommand{\vecc}{\vec{c}}

\newcommand{\vecs}{\vec{s}}
\newcommand{\setseed}{\set{s}}
\newcommand{\cit}[1]{\vec{{c}}^{(#1)}}
\newcommand{\Cit}[1]{\set{c}^{(#1)}}
\newcommand{\citidx}[1]{{c}^{(#1)}}
\newcommand{\step}{\gamma}
\newcommand{\stepit}[1]{\step^{(#1)}}
\newcommand{\proj}{p}
\newcommand{\dist}{d}
\newcommand{\setdist}[1]{D_{#1}}
\newcommand{\indicator}[1]{[#1]}
\renewcommand{\complement}[1]{\overline{#1}}

\newcommand{\emc}{\textsc{EMc}}
\newcommand{\pgd}{\textsc{PGDc}}
\newcommand{\pgds}[1]{\pgd-$#1$}
\newcommand{\pgdd}{\pgd-d}
\newcommand{\emcs}[1]{\emc-$#1$}
\newcommand{\emcd}{\emc-d}
\newcommand{\YL}{\textsc{YL}}
\newcommand{\HK}{\textsc{HK}}

\newcommand{\PPR}{\textsc{PPR}}

\jmlrheading{17}{2016}{1-28}{12/15; Revised 6/16}{8/16}{Twan van Laarhoven and Elena Marchiori}
\ShortHeadings{Local Network Community Detection with Continuous Optimization}{van Laarhoven and Marchiori}

\firstpageno{1}

\title{Local Network Community Detection with Continuous Optimization of Conductance and Weighted Kernel K-Means}

\author{%
  \name Twan van Laarhoven %
  \email tvanlaarhoven@cs.ru.nl
    \AND
  \name Elena Marchiori %
  \email elenam@cs.ru.nl \\
  \addr Institute for Computing and Information Sciences\\ Radboud University Nijmegen\\Postbus 9010\\6500 GL Nijmegen, The Netherlands %
}
\editor{Jure Leskovec}

\begin{document}

\maketitle

\begin{abstract}
Local network community detection is the task of finding a single community of nodes concentrated around few given seed nodes in a localized way. Conductance is a popular objective function used in many algorithms for local community detection.
This paper studies a continuous relaxation of conductance. We show that continuous optimization of this objective still leads to discrete communities.
We investigate the relation of conductance with weighted kernel k-means for a single community, which leads to the introduction of a new objective function, $\sigma$-conductance. Conductance is obtained by setting $\sigma$ to $0$.
Two algorithms, {\emc} and {\pgd}, are proposed to locally optimize $\sigma$-conductance and  automatically tune the parameter $\sigma$. They are  based on expectation maximization  and projected gradient descent, respectively. 
We prove locality and give performance guarantees for {\emc} and {\pgd} for a class of dense and well separated communities centered around the seeds.   
Experiments are conducted on networks with ground-truth communities, comparing to state-of-the-art graph diffusion algorithms  for conductance optimization.  
On large graphs,  results indicate that {\emc} and {\pgd}  stay localized and produce communities most similar to the ground, while graph diffusion  algorithms  generate  large communities of lower quality.\footnote{Source code of the algorithms used in the paper is available at 
\url{http://cs.ru.nl/~tvanlaarhoven/conductance2016}.}

\end{abstract}

\begin{keywords}
  community detection, stochastic block model, conductance
\end{keywords}

\section{Introduction}

Imagine that you are trying to find a community of nodes in a network around a given set of nodes. A simple way to approach this problem is to consider this set as seed nodes, and then keep adding nodes in a local neighborhood of the seeds as long as this makes the community better in some sense. In contrast to global clustering,  where the overall community structure of a network has to be found, local community detection aims to find only one community around the given seeds by relying on local computations involving only nodes relatively close to the seed. 
Local community detection by seed expansion is especially beneficial in large networks, and is commonly used in real-life large scale network analysis \citep{gargi2011large, leskovec2010empirical,wu2012learning}. 

Several algorithms for local community detection operate by seed expansion. These methods have different expansion strategies, but what they have in common is their use of conductance as the objective to be optimized. 
Intuitively, conductance measures how strongly a set of nodes is connected to the rest of the graph; sets of nodes that are isolated from the graph have low conductance and make good communities.

The problem of finding a set of minimum conductance in a graph is computationally intractable \citep{Chawla2005,vsima2006np}. 
As a consequence,  many heuristic and approximation algorithms for local community detection have been introduced (see references in the related work section).
In particular, effective algorithms for this task are based on the local graph diffusion method. A  graph diffusion vector $\vec{f}$ is an infinite series  $\vec{f} = \sum_{i=0}^{\infty} \alpha_i \vec{P}^i \vec{s}$, with diffusion coefficients  $\sum_{i=0}^{\infty} \alpha_i  =1$, seed nodes $\vec{s}$, and random walk transition matrix $\vec{P}$. Types of graph diffusion, such as personalized Page Rank \citep{Andersen2006} and Heat Kernel \citep{Chung2007}, are determined by the choice of the diffusion coefficients. 
In the diffusion method an approximation of $\vec{f}$ is computed. After dividing each vector component by the degree of the corresponding node, the nodes are  sorted   in descending order by their values in this vector. Next, the conductance of each prefix of the sorted list is computed and either the set of smallest conductance is selected, e.g. in \citep{Andersen2006} or a local optima of conductance along the prefix length dimension \citep{Yang2012} is considered. 

These algorithms optimize conductance along a single dimension, representing the order in which nodes are added by the algorithm.  
However this ordering is mainly related to the seed, and not directly to the objective that is being optimized. 
Algorithms for the direct optimization of conductance mainly operate in the discrete search space of communities, and locally optimize conductance by adding and/or removing one node. This amounts to fixing a specific neighborhood structure over communities where the neighbors of a community are only those communities which differ by the membership of a single node.  This is just one possible choice of community neighbor.
A natural way to avoid the problem of choosing a specific neighborhood structure is to use continuous rather than discrete optimization.
To do this, we need a continuous relaxation of conductance, extending the notion of communities to allow for fractional membership.
This paper investigates such a continuous relaxation, which leads to the following findings. 

\subsubsection{On Local Optima}
Although local optima of a continuous relaxation of conductance might at first glance have nodes with fractional memberships, somewhat surprisingly all strict local optima are discrete.
This means that continuous optimization can directly be used to find communities without fractional memberships.

\subsubsection{Relation with Weighted Kernel K-Means}
We unravel the relation between conductance  and weighted kernel k-means objectives using the  framework by \citet{Dhillon:2007}.
Since the aim is to  find only one community, we consider a slight variation with one mean, that is, with $k=1$.  This relation leads to the introduction of a new objective function for local community detection, called $\sigma$-conductance, which is the sum of conductance and a regularization term whose influence is controlled by a parameter $\sigma$. Interestingly, the choice  of $\sigma$ has a direct effect on the number of local optima of the function, where larger values of $\sigma$ lead to more local optima. In particular, we prove that for $\sigma > 2$ all discrete communities are local optima. As a consequence, due to the seed expansion approach, local optimization of $\sigma$-conductance favors smaller communities for larger  values of $\sigma$.

\subsubsection{Algorithms}
Local optimization of $\sigma$-conductance can be easily performed using the projected gradient descent method. We develop an  algorithm based on this method, called {\pgd}.
Motivated by the relation between conductance and k-means clustering, we introduce an Expectation-Maximization (EM) algorithm for $\sigma$-conductance optimization, called {\emc}. We show that for $\sigma=0$, this algorithm is almost identical to projected gradient descent with an infinite step size in each iteration. We then propose a heuristic procedure for choosing $\sigma$ automatically in these algorithms.

\subsubsection{Retrieving Communities}
We  give a theoretic characterization of a class of communities, called dense and isolated communities, for which {\pgd}  and {\emc} perform optimally. For this class of communities the algorithms exactly recover a community from the seeds.
We investigate the relation between this class of communities and  the notion of $(\alpha, \beta)$-cluster proposed by \citep{MishraSST08}  for social networks analysis. And we show that, while all maximal cliques in a graph are $(\alpha, \beta)$-clusters, they are not necessarily dense and isolated communities. We give a simple condition on  the degree of the nodes  of a community which guarantees that a dense and isolated community satisfying such condition is also an $(\alpha, \beta)$-cluster.


\subsubsection{Experimental Performance}
We use publicly available artificial and real-life network data with labeled ground-truth communities to assess the performance of {\pgd} and {\emc}.
Results of the two methods are very similar, with {\pgd} performing slightly better, while {\emc} is slightly faster.  
These results are compared with those obtained by three state-of-the-art algorithms for conductance optimization based on the local graph diffusion: the popular Personalized Page Rank (\PPR) diffusion algorithm by \citet{Andersen2006}, a more recent variant by \cite{Yang2012} (here called \YL), and the Heat Kernel (HK) diffusion algorithm by \cite{Kloster2014}. On large networks {\pgd} and {\emc}  stay localized and produce communities which are more faithful to the ground truth than those generated by the considered graph diffusion algorithms. {\PPR} and {\HK} produce much larger communities with a low conductance, while the {\YL} strategy outputs very small communities with a higher conductance.




\subsection{Related Work}
The enormous growth of network data from diverse disciplines such as social and information science and biology has boosted research on network community detection  (see for instance the overviews  by  \citet{Schaeffer2007} and \citet{Fortunato2010}).  Here we confine ourself to literature we consider to be relevant to the present work, namely local community detection by seed expansion, and review related work on  conductance as objective function and its local optimization. We also briefly review research on other objectives functions, and on properties of communities and of seeds.




\subsubsection{Conductance and Its Local Optimization}
Conductance has been largely used for network community detection.  
For instance \citet{leskovec2008statistical} introduced the notion of network community profile plot to measure the quality of a `best' community as a function of community size in a network.  They used conductance to measure the quality of a community and analyze   a large number of communities of different size scales in real-world social and information networks.


Direct conductance optimization was shown to favor  communities which are quasi-cliques \citep{Kang2011} or communities of large size which include irrelevant subgraphs  \citep{Andersen2006,Whang2013}.  

Popular algorithms for local community detection employ the local graph diffusion method to find a community with small conductance. 

Starting from the seminal work by \citet{Spielman2004}  various  algorithms for local community detection by seed expansion based on this approach have been proposed \citep{andersen2006local,avron2015community, Chung2007,Kloster2014,Zhu2013}. The theoretical analysis in these works is largely based on a mixing result which shows that a cut with small conductance can be found by simulating a random walk starting from a single node for sufficiently many steps \citep{lovasz1990mixing}.  This result is used to prove that if the seed is near to a set with small conductance then the result of the  procedure is a community with a related conductance, which is returned in time proportional to the volume of the community (up to a logarithmic factor). 

\citet{mahoney2012local} performed local community detection by modifying the spectral program used in standard global spectral clustering. Specifically the authors incorporated a bias towards a target region of seed nodes in the form of a constraint to force the solution to be well connected with or to lie near the seeds.  The degree of connectedness was specified by setting  a so-called correlation parameter. The authors showed  that the optimal solution of the resulting constrained optimization problem is a generalization of  Personalized PageRank \citep{Andersen2006}.



\subsubsection{Other Objectives}
Conductance is not the only objective function used in local community detection algorithms. Various other objective functions have been considered in the literature. For instance, \citet{chen2009local} proposed to use the ratio of the average internal and external degree of nodes in a community as objective function. \citet{clauset2005finding} proposed a local variant of modularity. \citet{Wu:2015} modified the classical density objective, equal to the sum of edges in the community divided by its size, by replacing the denominator with the sum of weights of the community nodes, where  the weight of a node quantifies its proximity to the seeds and is computed using a graph diffusion method.

A comparative experimental analysis of objective functions with respect to their experimental and theoretical properties was performed e.g. in \citep{Yang2012} and \citep{Wu:2015}, respectively. 

\subsubsection{Properties of Communities}
Instead of focusing on objective functions and methods for local community detection, other researchers investigated properties of communities. \citet{MishraSST08} focused on interesting classes of communities and algorithms for their exact retrieval. They defined the so called $(\alpha, \beta)$-communities and developed algorithms capable of retrieving this type of communities starting from a seed connected to a large fraction of the members of the community.
\citet{zhu2013local}  considered the class of well-connected communities, which have a better internal connectivity  than  conductance. Internal connectivity of a community is defined as  the inverse of the mixing time for a random walk on the subgraph induced by the community.  They showed that for well-connected communities, it is possible to provide an improved  performance guarantee, in terms of conductance of the output, for local community detection algorithms based on the diffusion method.
\citet{Gleich:2012}  investigated the utility of neighbors of the seed; in particular  they showed empirically that such neighbors form a `good' local community around the seed.  \citet{Yang2012} investigated properties  of ground truth communities in social, information and technological networks.  

\citet{Lancichinetti2011finding} addressed the problem of finding a significant local community from an initial group of nodes. They  proposed a method which locally optimizes the statistical significance of a community, defined with respect to a global null model, by iteratively adding external significant nodes and removing internal nodes that are not statistically relevant. The resulting community is not guaranteed to contain the nodes of the initial community. 

\subsubsection{Properties of Seeds}
Properties of seeds in relation to the  performance of algorithms were investigated by e.g. \citet{Kloumann2014}.  They considered different types of algorithms, in particular a greedy seed expansion algorithm which at each step adds the node that yields the most negative change in conductance \citep{Mislove:2010}.  \citet{Whang2013}  investigated various methods  for choosing the seeds for a PageRank based algorithm for community detection. \citet{chen2013detecting} introduced the notion of local degree central node,  whose degree is greater than or equal to the degree of its neighbor nodes. A new local community detection method is introduced based on the local degree central node. In this method, the local community is not discovered from the given starting node, but from the local degree central node that is associated with the given starting node. 
\subsection{Notation}

We start by introducing the notation used in the rest of this paper.
We denote by $\nodes$ the set of nodes in a network or graph $G$.
A community, also called a cluster, $\setc \subseteq \nodes$ will be a subset of nodes, and its complement $\complement{\setc} = \nodes \setminus \setc$ consists of all nodes not in $\setc$.
Note that we consider any subset of nodes to be a community, and the goal of community detection is to find a \emph{good} community. 

Let $A$ be the adjacency matrix of $G$, where $a_{ij}$ denotes the weight of an edge between nodes $i$ and $j$.
In unweighted graphs $a_{ij}$ is either 0 or 1, and in undirected graphs $a_{ij} = a_{ji}$.
In this paper we work only with unweighted undirected graphs.
We can generalize this notation to sets of nodes, and write $a_{xy} = \sum_{i \in x}\sum_{j \in y} a_{ij}$.
With this notation in hand we can write conductance as
\begin{equation*}
  \score(\setc) = \frac{a_{\setc\complement{\setc}}}{a_{\setc\nodes}} = 1 - \frac{a_{\setc\setc}}{a_{\setc\nodes}}
  .
\end{equation*}
A common alternative definition is
\begin{equation*}
  \score_\text{alt}(\setc) = \frac{a_{\setc\complement{\setc}}}{\min(a_{\setc\nodes},a_{\complement{\setc}\nodes})},
\end{equation*}
which considers the community to be the smallest of $\setc$ and $\complement{\setc}$.
%
For instance
\citet{Kloster2014} and \citet{Andersen2006} use this alternative definition, while \citet{Yang2012} use $\score$.

Note that $\phi$ has a trivial optimum when all nodes belong to the community, while $\phi_\text{alt}$ will usually have a global optimum with roughly half of the nodes belonging to the community. Neither of these optima are desirable for finding a single small community.


With a set $\set{x}$ we associate an indicator vector $\indicator{\set{x}}$ of length $|V|$, such that $\indicator{\set{x}}_i = 1$ if $i \in \set{x}$ and $\indicator{\set{x}}_i=0$ otherwise. We will usually call this vector $\vec{x}$.

\section{Continuous Relaxation of Conductance}

If we want to talk about directly optimizing conductance, then we need to define what (local) optima are.
The notion of local optima depends on the topology of the input space, that is to say, on what communities we consider to be neighbors of other communities. We could, for instance, define the neighbors of a community to be all communities that can be created by adding or removing a single node. But this is an arbitrary choice, and we could equally well define the neighbors to be all communities reached by adding or removing up to two nodes.
An alternative is to move to the continuous world, where we can use our knowledge of calculus to give us a notion of local optima.

To turn community finding into a continuous problem, instead of a set $\setc$ we need to see the community as a vector $\vec{c}$ of real numbers between $0$ and $1$, where $c_i$ denotes the degree to which node $i$ is a member of the community. Given a discrete community $\setc$, we have $\vecc = \indicator{\setc}$, but the inverse is not always possible, so the vectorial setting is more general.

The edge weight between sets of nodes can be easily generalized to the edge weight of membership vectors,
\[
  a_{\vec{x}\vec{y}} = \vec{x}\tr A \vec{y} = \sum_{i \in \nodes}\sum_{j \in \nodes} x_{i} a_{ij} y_{j}.
\]
Now we can reinterpret the previous definition of conductance as a function of real vectors, which we could expand as
\begin{equation*}
  \score(\vecc) = 1 - \frac{\sum_{i,j \in \nodes} c_i a_{ij} c_j}{\sum_{i,j \in \nodes} c_i a_{ij}}.
\end{equation*}
With this definition we can apply the vast literature on constrained optimization of differentiable functions.
In particular, we can look for local optima of the conductance, subject to the constraint that $0 \le c_i \le 1$.
These local optima will satisfy the Karush-Kuhn-Tucker conditions, which in this case amounts to, for all $i \in \nodes$,
\begin{align*}
  &0 \le c_i \le 1\\
  &\nabla\score(\vecc)_i \ge 0  \quad\text{ if }c_i = 0\\
  &\nabla\score(\vecc)_i = 0    \quad\text{ if }0 < c_i < 1,\\
  &\nabla\score(\vecc)_i \le 0  \quad\text{ if }c_i = 1.
\end{align*}

To use the above optimization problem for finding communities from seeds, we add one additional constraint.
Given a set $\setseed$ of seeds we require that $c_i \ge s_i$; in other words, that the seed nodes are members of the community.
This is the only way in which the seeds are used, and the only way in which we can use the seeds without making extra assumptions.

\subsection{A Look at the Local Optima}
\label{sec:optima-discrete}
By allowing community memberships that are real numbers, uncountably many more communities are possible. One might expect that it is overwhelmingly likely that optima of the continuous relaxation of conductance are communities with fractional memberships.
But this turns out not to be the case. In fact, the strict local optima will all represent discrete communities.

To see why this is the case, consider
the objective in terms of the membership coefficient $c_i$ for some node $i$. This takes the form of a quadratic rational function,
\newcommand{\coeff}{\alpha}
\[
  \score(c_i) = \frac{\coeff_1 + \coeff_2 c_i + \coeff_3 c_i^2}{\coeff_4 + \coeff_5 c_i}.
\]
The coefficients in the denominator are positive, which means that the denominator is also positive for $c_i>0$.
At an interior local minimum we must have $\score'(c_i)=0$, which implies that $\score''(c_i)=2\coeff_3/(\coeff_4+\coeff_5 c_i)^3$.
But $\coeff_3 \le 0$, since it comes from the $c_i a_{ii} c_i$ term in the numerator of the conductance, so $\score''(c_i)\le 0$, and hence there are only local maxima or saddle points, not strict local minima.


It is still possible for there to be plateaus in the objective functions, where $\score(\vec{c})$ is optimal regardless of the value of $c_i$ for a certain node $i$.


\subsection{The Relation to Weighted Kernel K-Means Clustering}


Another view on conductance is by the connection to weighted kernel $k$-means clustering.
The connection between weighted kernel $k$-means and objectives for graph partitioning has been thoroughly investigated in \cite{Dhillon:2007}.
Here we extend that connection to the single cluster case.

Start with weighted $k$-means clustering, which, given a dataset $\{x_i\}_{i=1}^N$ and weights $\{w_i\}$, minimizes the following objective
\begin{align*}
  \sum_{i=1}^N \sum_{j=1}^k w_i c_{ij} \| x_i - \mu_j \|_2^2
\end{align*}
with respect to $\mu_j$ and $c_{ij}$, where $c_{ij}$ indicates if point $i$ belongs to cluster $j$,
subject to the constraint that exactly one $c_{ij}$ is 1 for every $i$.

Since our goal is to find a single cluster, a first guess would be to take $k=2$, and to try to separate a foreground cluster from the background.
But when using $2$-means, there is no distinction between foreground and background, and so solutions will naturally have two clusters of roughly equal size.
Instead, we can consider a one-cluster variant that distinguishes between points in a cluster and background points, which we call $1$-mean clustering. This can be formulated as the minimization of
\begin{align*}
  \sum_i w_i \bigl(c_i \| x_i - \mu \|_2^2 + (1-c_i) \lambda_i\bigr)
\end{align*}
with respect to a single $\mu$ and cluster membership indicators $c_i$ (between 0 and 1). Here $\lambda_i$ is a cost for node $i$ being a member of the background.

We allow different $\lambda_i$ for different nodes, as there is no reason to demand a single value. The condition for a node $i$ to be part of the community is $\| x_i - \mu \|_2^2 < \lambda_i$. So different values for $\lambda_i$ might be useful for two reasons. The first would be to allow incorporating prior knowledge, the second reason would be if the scale (of the clusters) is different, that is, nodes (in different clusters) have different distances from the mean. By adding a diagonal matrix to the kernel, the squared distance from all points to all other points is increased by that same amount. It makes sense to compensate for this in the condition for community membership. And since the diagonal terms we add to the kernel vary per node, the amount that these nodes move away from other points also varies, which is why we use different $\lambda_i$ per node.


The minimizer for $\mu$ is the centroid of the points inside the cluster,
\begin{align*}
  \mu = \frac{\sum_{i} w_i c_{i} x_i}{\sum_{i} w_i c_{i}};
\end{align*}
while the minimizer for $c_i$ is 1 if and only if $\|x_i - \mu \|^2 < \lambda_i$, and $0$ otherwise.


The $k$-means and $1$-mean objectives can be kernelized by writing distances in terms of inner products, and using a kernel $K(i,j) = \inner{x_i}{x_j}$.
The cluster mean is then a linear combination of points, $\mu = \sum_i \mu_{i} x_i$, giving
\begin{align*}
  \| x_i - \mu \|_2^2 = K(i,i) - 2\sum_{j} \mu_{j} K(i,j) + \sum_{j,k} \mu_{j} K(j,k) \mu_{k}.
\end{align*}
By filling in the optimal $\mu$ given above, the $1$-mean objective then becomes
\begin{align*}
  \score_{W,K,\lambda}(\vecc) =
  & \sum_i w_i c_i (K(i,i) - \lambda_i)
    + \sum_i w_i \lambda_i
    \\&- \frac{\sum_{i,j} w_i c_i w_j c_j K(i,j)}{\sum_{i} w_i c_i}.
\end{align*}
The second term is constant, so we can drop it for the purposes of optimization.
%

We pick $\lambda_i = K(i,i)$. With this choice, the condition for a node $i$ to be a member of the community is $\|x_i - \mu\|_2^2 < \|x_i - 0\|_2^2$. This can be seen as a 2-means cluster assignment where the background cluster has the origin as the fixed mean. With this choice the first term also drops out.

%

\label{sec:sigma}

By converting the graph into a kernel with
\begin{align*}
  K = W^{-1} A W^{-1},
\end{align*}
where $W$ is a diagonal matrix with the weights $w_i$ on the diagonal,
we can obtain objectives like conductance and association ratio.
However this $K$ is not a legal kernel, because a kernel has to be positive definite.
Without a positive definite kernel the distances $\|x_i - \mu\|$ from the original optimization problem can become negative.
To make the kernel positive definite, we follow the same route as \citeauthor{Dhillon:2007}, and add a diagonal matrix, obtaining 
\begin{align*}
  K = \sigma W^{-1} + W^{-1} A W^{-1}.
\end{align*}


Since we are interested in conductance, we take as weights $w_i = a_{i\nodes}$, the degree of node $i$,
and we take $\lambda_i = K(i,i)$. 
This results (up to an additive constant) in the following objective which we call $\sigma$-conductance,
\begin{align*}
  \score_{\sigma}(\vecc) = 
    1
    - \frac{\sum_{i,j} c_i c_j a_{ij}}{\sum_{i} c_i a_{i\nodes}}
    - \sigma\frac{\sum_{i} c_i^2 a_{i\nodes}}{\sum_{i} c_i a_{i\nodes}}.
\end{align*}

Observe that if $\vecc$ is a discrete community, then $c_i^2=c_i$, and the last term is constant. In that case optimization of this objective is exactly equivalent to optimizing conductance.

For the purposes of continuous optimization however, increasing the $\sigma$ parameter has the effect of increasing the objective value of non-discrete communities. So different communities become more separated, and in the extreme case, every discrete community becomes a local optimum.

\begin{theorem}
  When $\sigma > 2$, all discrete communities $\vecc$ are local minima of $\score_\sigma(\vecc)$ constrained to $0 \le c_i \le 1$.
\end{theorem}
\begin{proof}
  The gradient of $\score_\sigma$ is
  \[
    \nabla\score_\sigma(\vecc)_i =
      a_{i\nodes}\frac{a_{\vecc\vecc}}{a_{\vecc\nodes}^2} - 2 \frac{a_{i\vecc}}{a_{\vecc\nodes}} +
      \sigma\Bigl(
       a_{i\nodes}\frac{\sum_j c_j^2 a_{j\nodes}}{a_{\vecc\nodes}^2} - 2 c_i \frac{a_{i\nodes}}{a_{\vecc\nodes}}
      \Bigr).
  \]
  When $\vecc$ is discrete, then $\sum_j c_j^2 a_{j\nodes} = a_{\vecc\nodes}$, so the gradient simplifies to
  \[
    \nabla\score_\sigma(\vecc)_i =
    \frac{a_{i\nodes}}{a_{\vecc\nodes}}
    \Bigl(
     \frac{a_{\vecc\vecc}}{a_{\vecc\nodes}} + (1 - 2c_i)\sigma - 2 \frac{a_{i\vecc}}{a_{i\nodes}}
    \Bigr)
      .
  \]
  Because $a_{i\vecc} \le a_{i\nodes}$ and $a_{\vecc\vecc} \le a_{\vecc\nodes}$ we can bound this by
  \[
    \frac{a_{i\nodes}}{a_{\vecc\nodes}}
    \bigl((1 - 2c_i)\sigma - 2\bigr)
    \le
    \nabla\score_\sigma(\vecc)_i
    \le
    \frac{a_{i\nodes}}{a_{\vecc\nodes}} \bigl((1 - 2 c_i)\sigma + 1\bigr).
  \]
  So if $c_i = 0$, we get that $\nabla\score_\sigma(\vecc)_i > 0$ when $\sigma > 2$.
  And if $c_i = 1$, we get that $\nabla\score_\sigma(\vecc)_i < 0$ when $\sigma > 1$.
  
  This means that when $\sigma > 2$ all discrete communities satisfy the KKT conditions,
  and from the sign of the gradient we can see that they are not local maxima.
  Furthermore, $\score_\sigma(c)$ is a concave function, so it has no saddle points (see the proof of \thmref{thm:optima-discrete}).
  This means that all discrete communities are local minima of $\score_\sigma$.
\end{proof}

Conversely, the result from \secref{sec:optima-discrete} generalizes to $\sigma$-conductance,
\begin{theorem}
  When $\sigma \ge 0$, all strict local minima $\vecc$ of $\score_\sigma(\vecc)$ constrained to $0 \le c_i \le 1$ are discrete.
  Furthermore, if $\sigma > 0$ then all local minima are discrete.
  \label{thm:optima-discrete}
\end{theorem}
\begin{proof}
  By the argument from \secref{sec:optima-discrete}.
  When $\sigma > 0$ it is always the case that $\coeff_3 < 0$, so there are no saddle points or plateaus, and all local minima are discrete.
\end{proof}

As an example application of $\sigma$-conductance, consider the network in \figref{fig:clique-with-tail}.
In this network, the clique is not a local optimum of regular conductance. This is because the gradient for the adjacent nodes with degree 2 is always negative, regardless of the conductance of the community. However, for $\sigma$-conductance this gradient becomes positive when $\sigma > \score_\sigma(\vecc)$, in this case when $\sigma > 0.131$.
In other words, with higher $\sigma$, adjacent nodes with low degree are no longer considered part of otherwise tightly connected communities such as cliques.

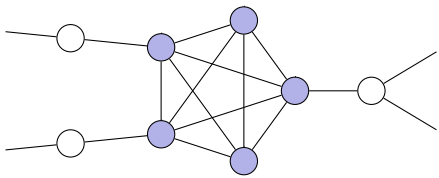
\begin{figure}
  \def\circlescale{0.82}
  \newcommand{\circleclique}[5]{
    \begin{scope}[#3]
      \foreach \x in {1,...,#2} {
        \node (#1\x) [#4] at (\x*360/#2:1.2*\circlescale) {};
        \foreach \y in {1,...,#2} {
          \ifnum \x=\y \breakforeach \fi
          \ifnumequal{\x}{#5}{}{
            \draw[] (#1\x) -- (#1\y);
          }
        }
      }
    \end{scope}
  }
  \begin{center}
    \begin{tikzpicture}
      [node/.style={circle,draw}
      ,cnode/.style={circle,draw,fill=blue!70!black!30}
      ]
      \circleclique{a}{5}{shift={(0,0)},rotate=0}{cnode}{0}
      \node[node] (1) at (2,0.) {};
      \node[] (2) at (3,0.6) {};
      \node[] (2b) at (3,-0.6) {};
      \node[node] (l1) at (-2,0.7) {};
      \node[] (l2) at (-3,0.8) {};
      \node[node] (k1) at (-2,-0.7) {};
      \node[] (k2) at (-3,-0.8) {};
      \draw[] (a5) -- (1) -- (2);
      \draw[] (1) -- (2b);
      \draw[] (a2) -- (l1) -- (l2);
      \draw[] (a3) -- (k1) -- (k2);
    \end{tikzpicture}
  \end{center}
  \caption{A simple subnetwork consisting of a clique with tails connecting it to the rest of the network. The clique (shaded nodes) is not a local optimum of conductance, but it is a local optimum of $\sigma$-conductance when $\sigma > 0.131$.}
  \label{fig:clique-with-tail}
\end{figure}

\section{Algorithms}

We now introduce two simple algorithms for the local optimization of conductance and $\sigma$-conductance, analyze their computational complexity and provide an exact performance guarantee for a class of communities.
Then we look at a procedure for the automatic selection of a value for $\sigma$.

\subsection{Projected Gradient Descent}
\label{sec:pgd}

%

Perhaps the simplest possible method for constrained continuous optimization problems is projected gradient descent.
This is an iterative algorithm, where in each step the solution is moved in the direction of the negative gradient, and then this solution is projected so as to satisfy the constraints.

In our case, we start from an initial community containing only the seeds,
\begin{align*}
  \cit{0} = \vecs,
\end{align*}
where $\vecs=[S]$ is a sparse vector indicating the seed node(s).
Then in each subsequent iteration we get
\begin{align*}
  \cit{t+1} = \proj( \cit{t} - \stepit{t}\nabla\score(\cit{t}) ).
\end{align*}
This process is iterated until convergence.
The step size $\stepit{t}$ can be found with line search.
The gradient $\nabla\score$ is given by
\[
  \nabla\score(\vecc)_i = \frac{a_{i\nodes}a_{\vecc\vecc}}{a_{\vecc\nodes}^2} - 2 \frac{a_{ic}}{a_{\vecc\nodes}}.
\]
And the projection $\proj$ onto the set of valid communities is defined by
\[
  \proj(\vecc) = \argmin_{{\vecc' \text{, s.t. }0 \le c_i' \le 1, s_i \le c_i'}} \|\vecc - \vecc'\|^2_2,
\]
which simply amounts to
\[
  \proj(\vecc) = \max(\vecs,\min(1,\vecc)).
\]
This function clips values above 1 to 1, and values below $s_i$ to $s_i$. Since $s_i \ge 0$ this also enforces that $c_i \ge 0$.

The complete algorithm is given in \algref{alg:pgd}.
If a discrete community is desired, as a final step, we might threshold the vector $\vecc$. But as shown in \thmref{thm:optima-discrete} the found community is usually already discrete.

\begin{algorithm}
  \caption{Projected Gradient Descent conductance optimization (\pgd)}
  \textbf{Input:} A set $\setseed$ of seeds of seeds, a graph $G$, a constant $\sigma \ge 0$.
  \begin{algorithmic}[1]
    \State $\vecs \gets \indicator{\setseed}$
    \State $\cit{0} \gets \vecs$
    \State $t \gets 0$
    \Repeat
      \State $\stepit{t} \gets \text{LineSearch}(\cit{t})$
      \State $\cit{t+1} = \proj(\cit{t} - \stepit{t} \nabla\score_\sigma(\cit{t}) )$
      \State $t \gets t+1$
    \Until{$\cit{t-1} = \cit{t}$}
    \State $\setc \gets \{ i \in \nodes \mid \cit{t}_i \ge 1/2 \}$
  \end{algorithmic}
  \vspace*{1mm}\hrule\vspace*{1mm}
  \textbf{function} {LineSearch}($\vecc$)
  \begin{algorithmic}[1]
      \State $\step^* \gets 0$, \quad $\score^* \gets \score_\sigma(\vecc)$
      \State $\vec{g} \gets \nabla\score_\sigma(\vecc)$
      \State $\step \gets 1/\max(|\vec{g}|)$
      \Repeat
        \State $\vecc' \gets \proj(\vecc - \step \vec{g})$
        \If{$\score_\sigma(\vecc') < \score^*$}
          \State $\step^* \gets \step$, \quad $\score^* \gets \score_\sigma(\vecc')$
        \EndIf
        \State $\step \gets 2\step$
      \Until{$\vecc'_i \in \{0,1\}$ for all $i$ with $g_i \neq 0$}
      \State \textbf{return} $\step^*$
  \end{algorithmic}
  \label{alg:pgd}
\end{algorithm}

\subsection{Expectation-Maximization}
\label{sec:em}


The connection to $k$-means clustering suggests that it might be possible to optimize conductance using an Expectation-Maximization algorithm similar to Lloyd's algorithm for $k$-means clustering.
Intuitively, the algorithm would work as follows:
\begin{compactitem}
 \item \textbf{E step} assign each node $i$ to the community if and only if its squared distance to the mean is less than $\lambda_i$.
 \item \textbf{M step} set the community mean to the weighted centroid of all nodes in the community.
\end{compactitem}
These steps are alternated until convergence.
Since both these steps do not increase the objective value, the algorithm is guaranteed to converge.

If the community after some iterations is $\set{c}$, then, as in the previous section, we can fill in the optimal mean into the E step, to obtain that a node $i$ should be part of the community if
\[
 K(i,i)
 +  \frac{a_{\setc\setc}/a_{\setc\nodes} + \sigma}{a_{\setc\nodes}}
 - 2\frac{a_{i\setc}/a_{i\nodes} + \sigma c_i}{a_{\setc\nodes}}
 < \lambda_i.
\]
When $\lambda_i = K(i,i)$, this condition is equivalent to
\[
  \nabla\score_\sigma(\set{c}) < 0.
\]
This leads us to the EM community finding algorithm, \algref{alg:em}.
\begin{algorithm}
  \caption{EM conductance optimization (\emc)}
  \textbf{Input:} A set $S$ of seeds, a graph $G$, a constant $\sigma \ge 0$.
  \begin{algorithmic}[1]
    \State $\Cit{0} \gets S$
    \State $t \gets 0$
    \RepeatWhile
      \State $\Cit{t+1} = \{i \mid \nabla\score_\sigma(\Cit{t})_i < 0\} \cup S$
      \State $t \gets t+1$
    \EndRepeatWhile{$\Cit{t} < \Cit{t-1}$}
  \end{algorithmic}
  \label{alg:em}
\end{algorithm}

By taking $\sigma=0$ we get that nodes are assigned to the community exactly if the gradient $\nabla\score(\set{c})_i$ is negative. So, this EM algorithm is very similar to projected gradient descent with an infinite step size in each iteration. The only difference is for nodes with $\nabla\score(\set{c})_i=0$, which in the {\emc} algorithm are always assigned to the background, while in PGD their membership of the community is left unchanged compared to the previous iteration.

Of course, we have previously established that $\sigma=0$ does not lead to a valid kernel (this doesn't preclude us from still using the EM algorithm). In the case that $\sigma>0$ there is an extra barrier for adding nodes not currently in the community, and an extra barrier for removing nodes that are in the community. This is similar to the effect that increasing $\sigma$ has on the gradient of $\score_{\sigma}$.

\subsection{Computational Complexity}


Both methods require the computation of the gradient in each iteration.
This computation can be done efficiently. The only nodes for which the gradient of the conductance is negative are the neighbors of nodes in the current community, and the only nodes for which a positive gradient can have an effect are those in the community. So the gradient doesn't need to be computed for other nodes. For the other nodes the gradient depends on the number of edges to the community, and on the node's degree. Assuming that the node degree can be queried in constant time, the total time per iteration is proportional to the size of the one-step-neighborhood of the community, which is of the order of the volume of the community. If the node degrees are not known, then the complexity  increases to be proportional to the volume of the one-step-neighborhood of the community, though this is a one-time cost, not a per iteration cost.

As seen in \secref{sec:recover}, for dense and isolated communities, the number of iterations is bounded by the diameter of the community. In general we can not guarantee such a bound, but in practice the number of iterations is always on the order of the diameter of the recovered community.

For very large datasets, the computation of the gradient can still be expensive, even though it is a local operation.
Therefore, we restrict the search to a set of $1000$ nodes near the seed.
This set $N$ is formed by starting with the seed, and repeatedly adding all neighbors of nodes in $N$, until the set would contain more than 1000 nodes. In this last step we only add the nodes with the highest $a_{iN}/a_{iV}$ so that the final set contains exactly 1000 nodes.

\subsection{Choosing $\sigma$}
\label{sec:choosing-sigma}

In \secref{sec:sigma} we introduced the $\sigma$ parameter, and we have shown that larger values of $\sigma$ lead to more local optima. This leaves the question of choosing the value of $\sigma$.

One obvious choice is $\sigma=0$, which means that $\score_\sigma$ is exactly the classical conductance.

Another choice would be to pick the smallest $\sigma$ that leads to a positive definite kernel.
But this is a global property of the network, that is furthermore very expensive to compute.


Instead, we try several different values of $\sigma$ for each seed, and then pick the community with the highest density, that is, the community $\setc$ with the largest $a_{\setc\setc}/|\setc|^2$.

\subsection{Exactly Recoverable Communities}
\label{sec:recover}

We now take a brief look at which kinds of communities can be exactly recovered with gradient descent and expectation-maximization.
Suppose that we wish to recover a community $\setctrue$ from a seeds set $\setseed$, and assume that this community is connected.
Denote by $\dist(i)$ the shortest path distance from a node $i\in\setctrue$ to any seed node, in the subnetwork induced by $\setctrue$.

First of all, since both algorithms grow the community from the seeds, we need to look at subcommunities $\setc \subseteq \setctrue$ \emph{centered around the seeds}, by which we mean that $\dist(i) \le \dist(j)$ for all nodes $i\in \setc$ and $j \in \setctrue \setminus \setc$.

Secondly, we need the community to be sufficiently densely connected to be considered a community in the first place; but at the same time the community needs to be separated from the rest of the network.
Again, because the communities are grown, we require that this holds also for subcommunities that are grown from the seeds,
\begin{definition}
  A community $\setctrue$ is \emph{dense and isolated} with threshold $\sigma$ if
  for all subsets $\setc\subseteq \setctrue$ centered around the seeds $\setseed$:
  \begin{itemize}
    \item $2a_{i\setc}/a_{i\nodes} > a_{\setc\setc}/a_{\setc\nodes}-\sigma$ for all nodes $i \in \setc$, and
    \item $2a_{i\setc}/a_{i\nodes} \le a_{\setc\setc}/a_{\setc\nodes}-\sigma$ for all nodes $i \notin \setctrue$.
  \end{itemize}
\end{definition}
Some examples of communities that satisfy this property are cliques and quasi-cliques that are only connected to nodes of high degree.

Now denote by $\setdist{n}$ the set of nodes $i$ in $\setctrue$ with $\dist(i)\le n$. Clearly $\setdist{0} = \setseed$, and because the community is connected there is some $n^*$ such that $\setdist{n^*} = \setctrue$.

We first look at the expectation-maximization algorithm.
\begin{theorem}
  If $\setctrue$ is dense and isolated, then the iterates of the {\emc} algorithm satisfy $\Cit{t} = \setdist{t}$.
  \label{thm:recover-em}
\end{theorem}
\begin{proof}
  The proof proceeds by induction.
  For $t=0$, the only nodes $i$ with $\dist(i)=0$ are the seeds, and $\Cit{0} = \setseed$ by definition.
  
  Now suppose that $\Cit{t} = \setdist{t}$.
  Then for any node $i$ there are three possibilities.
  \begin{itemize}
    \item $i \in \setdist{t+1}$; then because $\setctrue$ is dense and $\setdist{t}$ is centered around the seeds,
          $2a_{i\Cit{t}}/a_{i\nodes} > 1-\score_\sigma(\Cit{t})$.
          This implies that $\nabla\score_\sigma(\Cit{t})_i < 0$.
    \item $i \in \setctrue \setminus \setdist{t+1}$; then there are no edges from $\setdist{t}$ to $i$, since otherwise the shortest path distance from $i$ to a seed would be $t+1$. So $a_{i\cit{t}} = 0$, which implies that $\nabla\score_\sigma(\Cit{t})_i \ge 0$.
    \item $i \notin \setctrue$; then because $\setctrue$ is isolated,
          $2a_{i\Cit{t}}/a_{i\nodes} \le 1-\score_\sigma(\Cit{t})$, which implies that $\nabla\score_\sigma(\cit{t})_i \ge 0$.
  \end{itemize}
  hence $\nabla\score_\sigma(\Cit{t})_i < 0$ if $i \in \setdist{t+1}$, and $\nabla\score_\sigma(\Cit{t})_i \ge 0$ otherwise.
  This means that $\Cit{t+1} = \setdist{t+1}$.
\end{proof}

For the projected gradient descent algorithm from Section~\ref{sec:pgd} we have an analogous theorem,
\begin{theorem}
  If $\setctrue$ is dense and isolated, then the iterates of {\pgd} satisfy $\cit{t} = \indicator{\setdist{t}}$. 
  \label{thm:recover-pgd}
\end{theorem}


\begin{proof}
  The proof proceeds by induction, and is analogous to the proof of \thmref{thm:recover-em}.
  For $t=0$, the only nodes $i$ with $\dist(i)=0$ are the seeds, and $\cit{0} = \vec{s} = \indicator{\setseed}$ by definition.
  
  Now suppose that $\cit{t} = \indicator{\setdist{t}}$. We have already shown that
  $\nabla\score_\sigma(\Cit{t})_i < 0$ if and only if $i \in \setdist{t+1}$.
  This means that after projecting onto the set of valid communities, only the membership of nodes in $\setdist{t+1}$ can increase. Since nodes in $\setdist{t}$ already have membership $1$, and nodes not in $\setdist{t+1}$ already have membership $0$, they are not affected.
  
  \newcommand{\stepmax}{\step_\text{max}}
  Let $\stepmax = \max_{i \in \setdist{t+1}} -1/\nabla\score_\sigma(\cit{t})_i$.
  Clearly if $\stepit{t} \ge \stepmax$, then $\citidx{t}_i - \stepit{t}\nabla\score_\sigma(\cit{t})_i>1$ for all nodes $i \in \setdist{t+1}$, and hence $\proj( \cit{t} - \stepit{t}\nabla\score_\sigma(\cit{t}) ) = \indicator{\setdist{t+1}}$.
  So to complete the proof, we only need to show that the optimal step size found with line search is indeed (at least) $\stepmax$.

  Suppose that $\stepit{t} < \stepmax$ leads to the optimal conductance.
  Then there is a node $i \in \setdist{t+1}$ with fractional membership, $0 < \citidx{t+1}_i < 1$.
  By repeated application of \thmref{thm:optima-discrete} we know that there is a discrete community $\setc'$ with $\score_\sigma(\setc') = \score_\sigma(\cit{t+1})$, and furthermore $\score_\sigma(\setc' \setminus \{i\}) = \score_\sigma(\setc' \cup \{i\})$. The latter can only be the case if $\nabla\score_\sigma(\setc')_i = 0$.
  Because the only nodes whose membership has changed compared to $\cit{t}$ are those in $\setdist{t+1} \setminus \setdist{t}$, it follows that $\setc'$ contains all nodes with distance at most $t$ to the seeds, as well as some nodes with distance $t+1$ to the seeds. This means that $\setc'$ is centered around the seeds, and so $\nabla\score_\sigma(\setc')_i > 0$. This is a contradiction, which means that $\stepit{t} \ge \stepmax$ must be the optimum.
\end{proof}

As a corollary, since $\setdist{n^*+1}=\setdist{n^*}$, both the {\emc} and the {\pgd} algorithm will halt, and exactly recover $\setctrue$.


The notion of dense and isolated community is weakly related to that of $(\alpha, \beta)$-cluster \citep{MishraSST08} (without the technical assumption that each node has a self-loop): $C$  is an $(\alpha, \beta)$-cluster, with  $0\leq \alpha < \beta \leq 1$ if $a_{iC}\geq  \beta |C|$ for $i$ in $C$,  $a_{iC}\leq  \alpha |C|$ for $i$  outside $C$. 

The definition of dense and isolated community depends on the degree of the nodes while that of  $(\alpha, \beta)$-cluster does not.
As a consequence, not all maximal cliques of a graph  are in general dense and isolated communities while they are $(\alpha, \beta)$-clusters. For instance, a maximal clique  linked to an external isolated node, that is, a node of degree $1$, is not dense and isolated.

In general one can easily show that if $C$ is dense and isolated and  $$ \min_{i \in C} a_{iV} > \max_{i\not\in C, a_{iC}>0} a_{iV} $$  then $C$ is an $(\alpha, \beta)$-cluster with
\[
  \beta =  \frac{1-\score(C)}{2|C|} \min_{i \in C} a_{iV}
\]
and
\[
  \alpha = \frac{1-\score(C)}{2|C|} \max_{i\not\in C, a_{iC}>0} a_{iV}.
\]

%

\begin{table*}[tb]
  \centering
\begin{tabular}{lrrr@{\hspace*{5mm}}r@{\hspace*{3mm}}r@{\hspace*{3mm}}r}
\hlinetop
Dataset & \#node & \#edge & clus.c. & \#comm & $\overline{|\setc|}$ & $\overline{\score(\setc)}$\\
\hlinemid
LFR (om=1)  &    5000 &    25125 & 0.039 &    101 & 49.5 & 0.302\\
LFR (om=2)  &    5000 &    25123 & 0.021 &    146 & 51.4 & 0.534\\
LFR (om=3)  &    5000 &    25126 & 0.016 &    191 & 52.4 & 0.647\\
LFR (om=4)  &    5000 &    25117 & 0.015 &    234 & 53.4 & 0.717\\
Karate      &      34 &       78 & 0.103 &      2 & 17.0 & 0.141\\
Football    &     115 &      613 & 0.186 &     12 & 9.6 & 0.402\\
Pol.Blogs   &    1490 &    16715 & 0.089 &      2 & 745.0 & 0.094\\
Pol.Books   &     105 &      441 & 0.151 &      3 & 35.0 & 0.322\\
Flickr      &   35313 &  3017530 & 0.030 &    171 & 4336.1 & 0.682\\
Amazon      &  334863 &   925872 & 0.079 & 151037 & 19.4 & 0.554\\
DBLP        &  317080 &  1049866 & 0.128 &  13477 & 53.4 & 0.622\\
Youtube     & 1134890 &  2987624 & 0.002 &   8385 & 13.5 & 0.916\\
LiveJournal & 3997962 & 34681189 & 0.045 & 287512 & 22.3 & 0.937\\
Orkut       & 3072441 & 117185083 & 0.014 & 6288363 & 14.2 & 0.977\\
CYC/Gavin 2006 &    6230 &     6531 & 0.121 &    408 & 4.7 & 0.793\\
CYC/Krogan 2006 &    6230 &     7075 & 0.075 &    408 & 4.7 & 0.733\\
CYC/Collins 2007 &    6230 &    14401 & 0.083 &    408 & 4.7 & 0.997\\
CYC/Costanzo 2010 &    6230 &    57772 & 0.022 &    408 & 4.7 & 0.996\\
CYC/Hoppins 2011 &    6230 &    10093 & 0.030 &    408 & 4.7 & 0.999\\
CYC/all     &    6230 &    80506 & 0.017 &    408 & 4.7 & 0.905\\
\hlinebot
\end{tabular}
  \caption{
    Overview of the datasets used in the experiments.
    For each dataset we consider three different sets of communities.
  }
  \label{tbl:datasets}
\end{table*}

\section{Experiments}


To test the proposed algorithms, we assess their performance on various networks. We also perform experiments  on recent state-of-the-art algorithms based on the diffusion method which also optimize conductance.

\subsection{Algorithms}

Specifically, we perform a comparative empirical analysis of the following algorithms.
\begin{enumerate}
  \item {\pgd}. The projected gradient descent algorithm for optimizing $\sigma$-conductance given in \algref{alg:pgd}.
  We show the results for two variants: {\pgds{0}} with $\sigma=0$ and {\pgdd} where $\sigma$ is chosen to maximize the community's density as described in \secref{sec:choosing-sigma}.
  
  \item \emc.  The Expectation Maximization algorithm for optimizing $\sigma$-conductance described in \secref{alg:em}.
  We consider the variants \emcs{0} with $\sigma=0$ and {\emcd} where $\sigma$ is chosen automatically.
  
  \item \YL.  The algorithm by \citet{Yang2012} (with conductance as scoring function), based on the diffusion method. It  computes an approximation of the personalized Page Rank graph diffusion vector \citep{andersen2006local}. The values in this vector are divided by the degree of the corresponding nodes, and the nodes are sorted  in descending order by their values.
  The ranking induces a one dimensional search space of communities $C_k$, called a sweep, defined by the sequence of prefixes of the sorted list, that is,  the $k$ top ranked nodes, for  $k=1, \dotsc, |V|$.
  The smallest $k$ whose $C_k$ is a `local optimum' of conductance  is computed and $C_k$ is extracted. Local optima of conductance over the one dimensional space $C_1, C_2, \dots, C_{|V|}$  are computed using a heuristic.  For increasing $k = 1, 2, \dotsc$ $\score(C_k)$ is measured. When $\score(C_k)$  stops decreasing at $k^*$  this is a  `candidate point' for a local minimum. It becomes a selected local minimum  if $\score(C_k)$ keeps increasing after $k^*$ and eventually becomes higher than $\alpha \score(C_k)$, otherwise it is discarded. $\alpha= 1.2$ is shown to give good results and is also used in our experiments. \citet{Yang2012}  show that  finding the local optima of the sweep curve instead of the global optimum  gives a large improvement over previous  local spectral clustering methods by \citet{Andersen2006}  and by \citet{Spielman2004}.
  %
  %
  
  \item \HK.  The algorithm by \citet{Kloster2014}, also based on the diffusion method. 
   Here, instead of using the Personalized PageRank score, nodes are ranked based on a Heat Kernel diffusion score \citep{Chung2007}.
We use the implementation made available by \citet{Kloster2014}, which tries different values of the algorithm's parameters  $t$ and $\epsilon$, and picks the community with the highest conductance among them. The details are in section 6.2 of \citep{Kloster2014}.    
   Code is available at \url{https://www.cs.purdue.edu/homes/dgleich/codes/hkgrow}.
  
  
  \item \PPR. The pprpush algorithm by \citet{Andersen2006} based on the personalized Page Rank graph diffusion.
  Compared to {\YL} instead of finding a local optimum of the sweep, the method looks for a global optimum, and hence often finds larger communities.
  We use the implementation included with the {\HK} method.
\end{enumerate}


\subsection{Datasets}
\subsubsection{Artificial Datasets}

The first set of experiments we performed is on artificially generated networks with a known community structure.
We use the LFR benchmark \citep{Fortunato2008BenchmarkGraphs}. 
We used the parameter settings \texttt{N=5000 mu=0.3 k=10 maxk=50 t1=2 t2=1 minc=20 maxc=100 on=2500}, which means that the graph has 5000 nodes, and between 20 and 100 communities, each with between 10 and 50 nodes. Half of the nodes, 2500 are a member of multiple communities. We vary the overlap parameter (\texttt{om}), which determines how many communities these nodes are in. More overlap makes the problem harder.

\subsubsection{Social and Information Network Datasets with Ground Truth}

We use five social and information network datasets with ground-truth from the SNAP collection \citep{snapnets}.
These datasets are summarized in \tblref{tbl:datasets}.
For each dataset we list the number of nodes, number of edges and the clustering coefficient.
We consider all available ground truth communities with at least 3 nodes.

\citet{Yang2012} also defined a set of top 5000 communities for each dataset. These are communities with a high combined score for several community goodness metrics, among which is conductance. We therefore believe that communities in this set are biased to be more easy to recover by optimizing conductance, and therefore do not consider them here. Results with these top 5000 ground truth communities are available in tables~1--3 in the supplementary material
\footnote{The supplementary material is available from \url{http://cs.ru.nl/~tvanlaarhoven/conductance2016}}.

In addition to the SNAP datasets we also include the Flickr social network dataset \citep{Wang-etal12-flickr-dataset}.

\subsubsection{Protein Interaction Network Datasets}

We have also run experiments on protein interaction networks of yeast from the BioGRID database \citep{biogrid}.
This database curates networks from several different studies. We have constructed networks for \citet{Gavin2006}, \citet{Krogan2006}, \citet{Collins2007}, \citet{Costanzo2010}, \citet{Hoppins2011}, as well as a network that is the union of all interaction networks confirmed by physical experiments.

As ground truth communities we take the CYC2008 catalog of protein complexes for each of the networks \citep{CYC2008}.

\subsubsection{Other Datasets}

Additionally we used some classical datasets with known communities:
Zachary's karate club \cite{Zachary1977};
Football: A network of American college football games \citep{GirvanNewman2002};
Political books: A network of books about US politics \citep{Krebs2004polbooks}; and
Political blogs: Hyperlinks between weblogs on US politics \citep{Adamic2005polblogs}.
These datasets might not be very well suited for this problem, since they have very few communities.

\subsection{Results}

In all our experiments we use a single seed node, drawn uniformly at random from the community.
We have also performed experiments with multiple seeds; the results of those experiments can be found in the supplementary material.

To keep the computation time manageable we have performed all experiments on a random sample of 1000 ground-truth communities.
For datasets with fewer than 1000 communities, we include the same community multiple times with different seeds.


Since the datasets here  considered have information about ground truth communities, a natural external validation criterion to assess the performance of algorithms on these datasets is to compare the community  produced by an algorithm with the ground truth one. In general, that is, when ground truth information is not available, this task is more subtle, because it is not clear what is a good external validation metric to evaluate a community \citep{Yang2015}.

\begin{table*}[t!]
  \centering
\begin{tabular}{lccccccc}
\hlinetop
Dataset & {\pgds{0}} & {\pgdd} & {\emcs{0}} & {\emcd} & {\YL} & {\HK} & {\PPR} \\
\hlinemid
LFR (om=1) & \textbf{0.967} & 0.185 & 0.868 & 0.187 & 0.203 & 0.040 & 0.041\\
LFR (om=2) & \textbf{0.483} & 0.095 & 0.293 & 0.092 & 0.122 & 0.039 & 0.041\\
LFR (om=3) & \textbf{0.275} & 0.085 & 0.158 & 0.083 & 0.110 & 0.037 & 0.039\\
LFR (om=4) & \textbf{0.178} & 0.074 & 0.100 & 0.072 & 0.092 & 0.032 & 0.034\\
Karate & 0.831 & 0.472 & 0.816 & 0.467 & 0.600 & 0.811 & \textbf{0.914}\\
Football & \textbf{0.792} & \textbf{0.816} & 0.766 & \textbf{0.805} & \textbf{0.816} & 0.471 & 0.283\\
Pol.Blogs & \textbf{0.646} & 0.141 & \textbf{0.661} & 0.149 & 0.017 & \textbf{0.661} & 0.535\\
Pol.Books & 0.596 & 0.187 & 0.622 & 0.197 & 0.225 & \textbf{0.641} & \textbf{0.663}\\
Flickr & 0.098 & 0.027 & 0.097 & 0.027 & 0.013 & 0.054 & \textbf{0.118}\\
Amazon & 0.470 & \textbf{0.522} & 0.425 & \textbf{0.522} & 0.493 & 0.245 & 0.130\\
DBLP & \textbf{0.356} & \textbf{0.369} & 0.317 & \textbf{0.371} & 0.341 & 0.214 & 0.210\\
Youtube & 0.089 & \textbf{0.251} & 0.073 & \textbf{0.248} & 0.228 & 0.037 & 0.071\\
LiveJournal & 0.067 & \textbf{0.262} & 0.059 & \textbf{0.259} & 0.183 & 0.035 & 0.049\\
Orkut & 0.042 & \textbf{0.231} & 0.033 & \textbf{0.231} & 0.171 & 0.057 & 0.033\\
CYC/Gavin 2006 & 0.474 & \textbf{0.543} & 0.455 & \textbf{0.543} & \textbf{0.526} & 0.336 & 0.294\\
CYC/Krogan 2006 & 0.410 & \textbf{0.513} & 0.364 & \textbf{0.511} & \textbf{0.504} & 0.229 & 0.169\\
CYC/Collins 2007 & 0.346 & \textbf{0.429} & 0.345 & \textbf{0.429} & \textbf{0.416} & 0.345 & 0.345\\
CYC/Costanzo 2010 & 0.174 & \textbf{0.355} & 0.172 & \textbf{0.351} & 0.314 & 0.170 & 0.170\\
CYC/Hoppins 2011 & 0.368 & \textbf{0.405} & 0.368 & \textbf{0.405} & \textbf{0.424} & 0.368 & 0.368\\
CYC/all & 0.044 & \textbf{0.459} & 0.017 & \textbf{0.459} & 0.425 & 0.016 & 0.002\\
\hlinebot
\end{tabular}
  \caption{
    Average $F_1$ score between recovered communities and ground-truth.
    The best result for each dataset is indicated in bold, as are the results not significantly worse according to a paired T-test (at significance level $0.01$).
  }
  \label{tbl:f1}
\end{table*}
\begin{table*}[t!]
  \centering
\begin{tabular}{lrrrrrrr}
\hlinetop
Dataset & {\pgds{0}} & {\pgdd} & {\emcs{0}} & {\emcd} & {\YL} & {\HK} & {\PPR} \\
\hlinemid
LFR (om=1) & 52.3 & 6.2 & 71.8 & 6.3 & 5.9 & 2410.0 & 2366.2\\
LFR (om=2) & 93.2 & 5.6 & 292.8 & 6.4 & 4.9 & 2404.4 & 2311.6\\
LFR (om=3) & 104.7 & 5.6 & 451.5 & 6.4 & 5.0 & 2399.4 & 2283.7\\
LFR (om=4) & 108.9 & 4.9 & 530.2 & 5.7 & 4.9 & 2389.1 & 2262.0\\
Karate & 20.0 & 8.0 & 24.1 & 8.0 & 8.8 & 16.7 & 17.1\\
Football & 14.7 & 9.5 & 16.2 & 9.4 & 8.8 & 40.5 & 56.5\\
Pol.Blogs & 515.6 & 110.1 & 538.9 & 118.1 & 7.2 & 492.7 & 1051.1\\
Pol.Books & 37.8 & 5.2 & 43.1 & 5.7 & 6.7 & 49.3 & 53.4\\
Flickr & 639.9 & 73.0 & 644.6 & 73.5 & 12.9 & 174.2 & 1158.1\\
Amazon & 25.2 & 5.6 & 45.6 & 5.7 & 6.4 & 88.8 & 20819.9\\
DBLP & 61.9 & 5.4 & 83.5 & 6.1 & 6.0 & 55.0 & 24495.0\\
Youtube & 340.6 & 19.4 & 474.2 & 21.3 & 9.3 & 147.9 & 20955.5\\
LiveJournal & 243.7 & 5.5 & 309.3 & 5.9 & 10.8 & 153.2 & 3428.7\\
Orkut & 245.1 & 17.9 & 344.7 & 19.1 & 11.1 & 212.0 & 1634.0\\
CYC/Gavin 2006 & 19.7 & 3.4 & 34.3 & 3.5 & 3.1 & 236.7 & 621.9\\
CYC/Krogan 2006 & 48.4 & 7.0 & 138.8 & 9.4 & 3.7 & 723.8 & 1756.3\\
CYC/Collins 2007 & 202.9 & 19.5 & 207.2 & 19.5 & 2.6 & 192.0 & 189.4\\
CYC/Costanzo 2010 & 540.2 & 58.1 & 564.2 & 66.5 & 5.9 & 1058.8 & 942.9\\
CYC/Hoppins 2011 & 229.9 & 110.1 & 235.5 & 110.4 & 4.3 & 295.2 & 295.2\\
CYC/all & 657.5 & 16.0 & 841.9 & 17.0 & 9.6 & 2795.2 & 5786.0\\
\hlinebot
\end{tabular}
  \caption{
    Average size of the recovered communities.
  }
  \label{tbl:cluster-size}
\end{table*}
\begin{table*}[t!]
  \centering
\begin{tabular}{lccccccc}
\hlinetop
Dataset & {\pgds{0}} & {\pgdd} & {\emcs{0}} & {\emcd} & {\YL} & {\HK} & {\PPR} \\
\hlinemid
LFR (om=1) & 0.301 & 0.750 & 0.304 & 0.749 & 0.755 & \textbf{0.250} & 0.273\\
LFR (om=2) & 0.532 & 0.786 & 0.541 & 0.787 & 0.780 & \textbf{0.315} & 0.338\\
LFR (om=3) & 0.589 & 0.793 & 0.587 & 0.793 & 0.781 & \textbf{0.333} & 0.354\\
LFR (om=4) & 0.604 & 0.791 & 0.595 & 0.792 & 0.775 & \textbf{0.341} & 0.359\\
Karate & 0.129 & 0.460 & \textbf{0.081} & 0.475 & 0.327 & 0.222 & 0.136\\
Football & 0.277 & 0.356 & 0.274 & 0.362 & 0.385 & 0.244 & \textbf{0.155}\\
Pol.Blogs & 0.228 & 0.743 & 0.212 & 0.737 & 0.867 & 0.229 & \textbf{0.137}\\
Pol.Books & 0.140 & 0.622 & 0.107 & 0.611 & 0.571 & 0.127 & \textbf{0.065}\\
Flickr & \textbf{0.777} & 0.937 & \textbf{0.777} & 0.937 & 0.951 & 0.864 & \textbf{0.762}\\
Amazon & 0.181 & 0.464 & 0.180 & 0.463 & 0.402 & 0.081 & \textbf{0.053}\\
DBLP & 0.246 & 0.571 & 0.257 & 0.565 & 0.498 & \textbf{0.133} & 0.147\\
Youtube & 0.601 & 0.765 & 0.711 & 0.759 & 0.700 & \textbf{0.201} & 0.341\\
LiveJournal & 0.563 & 0.875 & 0.589 & 0.874 & 0.774 & \textbf{0.336} & 0.489\\
Orkut & \textbf{0.718} & 0.916 & 0.731 & 0.917 & 0.928 & 0.750 & \textbf{0.711}\\
CYC/Gavin 2006 & 0.614 & 0.734 & 0.611 & 0.732 & 0.735 & \textbf{0.532} & \textbf{0.500}\\
CYC/Krogan 2006 & 0.466 & 0.626 & 0.469 & 0.620 & 0.617 & 0.325 & \textbf{0.265}\\
CYC/Collins 2007 & \textbf{0.716} & 0.953 & \textbf{0.712} & 0.953 & 0.972 & \textbf{0.720} & \textbf{0.730}\\
CYC/Costanzo 2010 & 0.759 & 0.931 & 0.755 & 0.929 & 0.934 & \textbf{0.672} & \textbf{0.646}\\
CYC/Hoppins 2011 & \textbf{0.788} & 0.883 & \textbf{0.785} & 0.882 & 0.970 & \textbf{0.763} & \textbf{0.763}\\
CYC/all & 0.674 & 0.872 & 0.742 & 0.874 & 0.840 & 0.363 & \textbf{0.026}\\
\hlinebot
\end{tabular}
  \caption{
    Average conductance of the recovered communities.
  }
  \label{tbl:conductance}
\end{table*}

\begin{figure}
  \centering
  \begin{tikzpicture}
    \begin{axis}[
        width=0.7\linewidth,height=0.55\linewidth,
        xlabel={$\sigma$},ylabel={$F_1$},
        minor tick num=1,
        xmin=0,xmax=1.07,ymin=0,ymax=0.61,enlargelimits=false,
        cycle list name=datasets,
        mark repeat=1,
        legend pos=outer north east, legend cell align=left, legend style={font=\small}
        ]
\addplot+[] coordinates{(0.000000,0.355800) (0.050000,0.399609) (0.100000,0.414009) (0.150000,0.411635) (0.200000,0.394356) (0.250000,0.376406) (0.300000,0.348982) (0.350000,0.312628) (0.400000,0.261248) (0.450000,0.257541) (0.500000,0.229172) (0.550000,0.227514) (0.600000,0.227205) (0.650000,0.224550) (0.700000,0.210215) (0.750000,0.209859) (0.800000,0.209395) (0.850000,0.209118) (0.900000,0.209169) (0.950000,0.208647) (1.000000,0.199739) (1.050000,0.199739) (1.100000,0.199739) (1.150000,0.199739) (1.200000,0.199739) (1.250000,0.199739) (1.300000,0.199739) (1.350000,0.199739) (1.400000,0.199739) (1.450000,0.199739) (1.500000,0.199739) (1.550000,0.199739) (1.600000,0.199739) (1.650000,0.199739) (1.700000,0.199739) (1.750000,0.199739) (1.800000,0.199739) (1.850000,0.199739) (1.900000,0.199739) (1.950000,0.199739) (2.000000,0.199206) };
addlegendentry{DBLP}
      %
\addplot+[] coordinates{(0.000000,0.470247) (0.050000,0.495688) (0.100000,0.513903) (0.150000,0.527889) (0.200000,0.535385) (0.250000,0.530863) (0.300000,0.522499) (0.350000,0.506155) (0.400000,0.463828) (0.450000,0.461005) (0.500000,0.416160) (0.550000,0.409687) (0.600000,0.406002) (0.650000,0.400886) (0.700000,0.366486) (0.750000,0.366232) (0.800000,0.366267) (0.850000,0.365322) (0.900000,0.365354) (0.950000,0.364430) (1.000000,0.346855) (1.050000,0.344757) (1.100000,0.344757) (1.150000,0.344757) (1.200000,0.344757) (1.250000,0.344757) (1.300000,0.344757) (1.350000,0.344757) (1.400000,0.344757) (1.450000,0.344757) (1.500000,0.344757) (1.550000,0.344757) (1.600000,0.344757) (1.650000,0.344757) (1.700000,0.344757) (1.750000,0.344757) (1.800000,0.344757) (1.850000,0.344757) (1.900000,0.344757) (1.950000,0.344757) (2.000000,0.341342) };
addlegendentry{Amazon}
      %
\addplot+[] coordinates{(0.000000,0.089102) (0.050000,0.182247) (0.100000,0.191400) (0.150000,0.201947) (0.200000,0.211479) (0.250000,0.216549) (0.300000,0.216883) (0.350000,0.219052) (0.400000,0.224307) (0.450000,0.224561) (0.500000,0.230972) (0.550000,0.229599) (0.600000,0.229850) (0.650000,0.231430) (0.700000,0.241004) (0.750000,0.241109) (0.800000,0.241016) (0.850000,0.240902) (0.900000,0.240584) (0.950000,0.239967) (1.000000,0.255225) (1.050000,0.255156) (1.100000,0.255156) (1.150000,0.255156) (1.200000,0.255156) (1.250000,0.255156) (1.300000,0.255156) (1.350000,0.255156) (1.400000,0.255156) (1.450000,0.255156) (1.500000,0.255156) (1.550000,0.255156) (1.600000,0.255156) (1.650000,0.255156) (1.700000,0.255156) (1.750000,0.255156) (1.800000,0.255156) (1.850000,0.255156) (1.900000,0.255156) (1.950000,0.255156) (2.000000,0.324277) };
addlegendentry{Youtube}
      %
\addplot+[] coordinates{(0.000000,0.067490) (0.050000,0.113929) (0.100000,0.141112) (0.150000,0.162706) (0.200000,0.185759) (0.250000,0.203251) (0.300000,0.215998) (0.350000,0.226733) (0.400000,0.239006) (0.450000,0.239079) (0.500000,0.254950) (0.550000,0.255206) (0.600000,0.255601) (0.650000,0.256086) (0.700000,0.272572) (0.750000,0.272723) (0.800000,0.272901) (0.850000,0.272753) (0.900000,0.272798) (0.950000,0.272886) (1.000000,0.294033) (1.050000,0.294033) (1.100000,0.294033) (1.150000,0.294033) (1.200000,0.294033) (1.250000,0.294033) (1.300000,0.294033) (1.350000,0.294033) (1.400000,0.294033) (1.450000,0.294033) (1.500000,0.294033) (1.550000,0.294033) (1.600000,0.294033) (1.650000,0.294033) (1.700000,0.294033) (1.750000,0.294033) (1.800000,0.294033) (1.850000,0.294033) (1.900000,0.294033) (1.950000,0.294033) (2.000000,0.318197) };
addlegendentry{LiveJournal}
      %
\addplot+[] coordinates{(0.000000,0.041593) (0.050000,0.129796) (0.100000,0.256721) (0.150000,0.307451) (0.200000,0.332399) (0.250000,0.343605) (0.300000,0.347253) (0.350000,0.350179) (0.400000,0.352398) (0.450000,0.352462) (0.500000,0.354006) (0.550000,0.354006) (0.600000,0.354006) (0.650000,0.354006) (0.700000,0.355276) (0.750000,0.355276) (0.800000,0.355276) (0.850000,0.355276) (0.900000,0.355276) (0.950000,0.355276) (1.000000,0.356345) (1.050000,0.356345) (1.100000,0.356345) (1.150000,0.356345) (1.200000,0.356345) (1.250000,0.356345) (1.300000,0.356345) (1.350000,0.356345) (1.400000,0.356345) (1.450000,0.356345) (1.500000,0.356345) (1.550000,0.356345) (1.600000,0.356345) (1.650000,0.356345) (1.700000,0.356345) (1.750000,0.356345) (1.800000,0.356345) (1.850000,0.356345) (1.900000,0.356345) (1.950000,0.356345) (2.000000,0.357790) };
addlegendentry{Orkut}
    \end{axis}
  \end{tikzpicture}
  \caption{Average $F_1$ score as a function of the $\sigma$ parameter on the SNAP datasets with the {\pgd} method.}
  \label{fig:sigma-f1}
\end{figure}
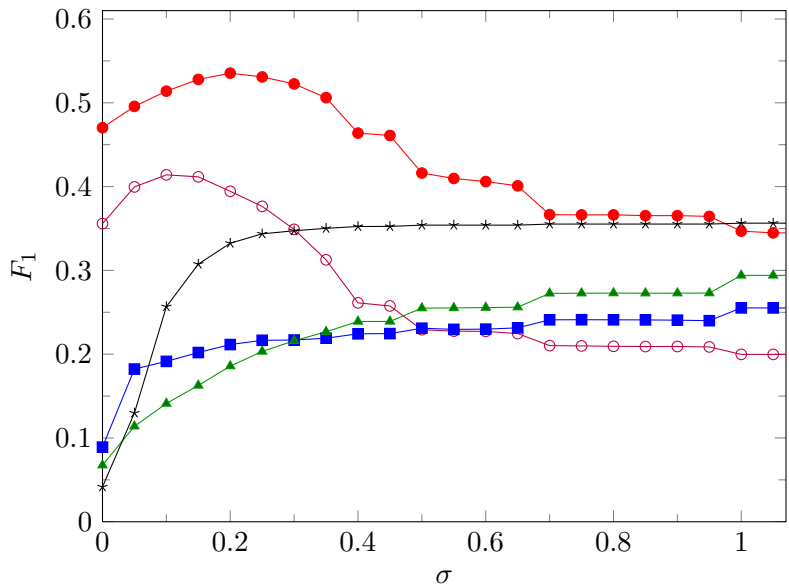

We measure quality performance with the $F_1$ score, which for community finding can be defined as
\[
  F_1(\setc,\setctrue) 
           = 2\frac{|\setc\cap \setctrue|}{|\setc|+|\setctrue|},
\]
where $\setc$ is the recovered community and $\setctrue$ is the ground truth one.
A higher $F_1$ score is better, with $1$ indicating a perfect correspondence between the two communities.

Note that a seed node might be in multiple ground truth communities.
In this case we only compare the recovered community to the true community that we started with. If a method finds another ground truth community this is not detected, and so it results in a low $F_1$ score.

We also analyze the output of these algorithms with respect to the conductance of produced communities and  their size.
Results on the run time of the algorithms are reported in the supplementary material (Table~4).

\Figref{fig:sigma-f1} shows the $F_1$ scores as a function of the parameter $\sigma$.
\Tblref{tbl:f1} shows the $F_1$ scores comparing the results of the methods to the true communities.
\Tblref{tbl:cluster-size} shows the mean size of the found communities, and \tblref{tbl:conductance} their conductance.

In general, results of these experiments indicate that on real-life networks, our methods based on continuous relaxation of conductance, {\PPR} and {\HK} produce communities with good conductance, but all are less faithful to the ground truth when the network contains many small communities.  In {\pgd}, {\emc} the automatic choice of $\sigma$ helps to achieve results closer to the ground truth, and the built-in tendency of {\YL} to favor small communities helps as well. 
On the other hand, on networks with large communities our methods, {\PPR} and {\HK}  work best.
On the artificial LFR data continuous relaxation of conductance seems to work best. This result indicates that the LFR model of `what is a community' is somehow in agreement with the notion of local community as local optimum of the continuous relaxation of conductance. However, as observed in  recent works like \citep{jeub2015think},  the LFR model does not seem to represent the diverse characteristics of real-life communities.

We have included tables of the standard deviation in the supplementary material.
Overall, the standard deviation in cluster size is of the same order of magnitude as the mean. The standard deviation of the conductance is around $0.1$ for LFR datasets, $0.2$ for the SNAP datasets and $0.3$ for the CYC datasets. It is not surprising that the variance is this high, because the communities vary a lot in size and density.

Results on these datasets can be summarized as follows. 

%
%
%

\subsubsection{Artificial LFR Datasets}
On these datasets, {\HK} and {\PPR} tend to find communities that are much too large, with small conductance but also with low $F_1$ scores.
This happens because the LFR networks are small, and the methods are therefore able to consider a large part of the nodes in the network.

On the other hand, {\YL} always starts its search at small communities, and it stops early, so the communities it finds are smaller than the ground truth ones on these networks.

The best $F_1$ results are achieved by {\pgd} with $\sigma=0$, that is, when the continuous relaxation of conductance is used as the objective function. This method employs a more powerful optimizer than YL, so it is able to find a large community with a better conductance, but it still stops at the first local optimum. In the LFR datasets these optima are very clear, and correspond closely to the ground truth communities.

In all cases {\emc} shows similar or slightly worse performance compared to {\pgd}, so the gradient descend algorithm should be preferred.

The automatic choice of $\sigma$ leads to communities which are of relatively small size. We believe that this happens because the nodes in LFR datasets all have exactly the same fraction of within community edges. Increasing $\sigma$ suddenly makes the gradient for most of these nodes positive. In real networks there are often hubs that are more central to a community, with more connections to the community's nodes and to the seed. These hubs still can be found at higher values of $\sigma$.


\subsubsection{Small Real-Life Social Networks with Few Communities (Karate, Football, Blogs, Books)}  
Our methods based on continuous relaxation of conductance yield the best $F_1$ results on the Football and Blog networks, while {\PPR} performs best on the other two networks and achieves best overall conductance. 

\subsubsection{Large Social Network with Big Communities (Flickr)}
On this network, {\PPR} achieves the best results both in terms of $F_1$ score as well as conductance. However the the produced communities are about four time smaller than the ground truth communities, which have more than $4000$ nodes.
{\pgd} and {\emc} with $\sigma=0$ yield communities of  conductance similar to that of {\PPR} communities, but their  size is smaller (about $650$ nodes).
This happens because the algorithms are restricted to $1000$ nodes around the seed, without this restriction larger communities would be found.
Somewhat surprisingly {\HK} produces communities of relatively small size (about $175$ nodes).
The automatic choice of $\sigma$ yields to even smaller communities (about $64$ nodes). The smallest size communities are produced by {\YL} (about $13$ nodes).

\subsubsection{Large Real-Life SNAP Networks with Many Small Communities}
On these networks the automatic choice of $\sigma$ gives best results, consistently outperforming the other algorithms. In \figref{fig:sigma-f1}, the  $F_1$ score of   {\pgd} and {\emc} as a function of $\sigma$ is plotted. For some datasets a small value of $\sigma$ works well, while for others a larger value of $\sigma$ is better.  Our procedure to choose $\sigma$ produces results that are close to, but slightly below, the best a posteriori choice of $\sigma$.  So on these networks the proposed procedure positively affects  the performance of our algorithms.
{\YL} favors communities of small size less faithful to the ground truth. {\pgds{0}}, {\emcs{0}}, {\PPR} and {\HK} `explode', and produce very large communities. For our methods this `explosion' is limited only because we limit the search to 1000 nodes near the seed.
The ground-truth communities of these datasets have rather high conductance, and the networks have  a very low clustering coefficient.
In such a case, communities have many links to nodes outside, hence conductance alone is clearly not suited to finding these type of local communities.

\subsubsection{Real-Life Protein Interaction Networks with Very Small Communities (CYC)}
Also on these networks the automatic choice of $\sigma$ gives best results. As expected, due to the very small size of the ground truth communities, {\YL} also achieves very good results. The other algorithms tend to produce less realistic,  large communities which have better conductance.

\subsubsection{Running Time}
The best performing algorithm with respect to running time is {\HK}.
{\pgd} and {\emc} are about four times slower with a fixed value of $\sigma$, and ten to twenty times slower when automatically determining $\sigma$.
The running times are included in Table~4 of the supplementary material. All experiments were run on a 2.4GHz Intel XEON E7-4870 machine.
Note that the different methods are implemented in different languages (our implementation is written in Octave, while {\HK} and {\PPR} are implemented in C++), so the running times only give an indication of the overall trend, and can not be compared directly.


\subsubsection{Top 5000 Communities}
Results with only the top 5000 ground truth communities available at the SNAP dataset collection  are similar to the results with all communities.
As expected, the $F_1$ score is much higher and the conductance of the recovered community is better.
Because these ground truth communities have a better conductance, it is better to optimize conductance, that is to take $\sigma=0$. As a consequence the performance of {\pgds{0}} and {\emcs{0}} is better than that of {\pgdd} and {\emcd} for these communities.
%
The full results are available in Tables~1--3 in the supplementary material.

\section{Discussion}

This paper investigated conductance as an objective function for local community detection from a set of seeds.
By making a continuous relaxation of conductance we show how standard techniques such as projected gradient descent can be used to optimize it.
Even though this is a continuous optimization problem, we show that the local optima are almost always discrete communities.
We further showed how linking conductance with kernel weighted $k$-means clustering leads to the new $\sigma$-conductance objective function and to simple yet effective algorithms for local community detection by seed expansion.   

We provided a formalization of a class of good local communities around a set of seeds and showed that the proposed algorithms can find them.  We suspect that these communities can also be exactly retrieved using local community algorithms based on the diffusion method, but do not yet have a proof.
The condition that such communities should be centered around the seeds raises the question of how to find such seeds. Although various works have studied seed selection for diffusion  based algorithms, such as \cite{Kloumann2014},  this problem remains to be investigated in the context of local community detection by $\sigma$-conductance optimization using {\pgd} and {\emc}.

Our experimental results indicate the effectiveness of direct optimization of a continuous relaxation of $\sigma$-conductance using gradient descent and expectation maximization. In our algorithms we used community density as a criterion to choose $\sigma$. 
This resulted to be a good choice for the performance of our algorithms on the SNAP networks. It would be interesting to investigate also other criteria to choose $\sigma$.
Conversely, the fact that maximum density is a good criterion for selecting $\sigma$ implies that it might also be directly optimized as an objective for finding communities.

On some datasets, when optimizing normal conductance, that is, with $\sigma=0$, our methods sometimes find very large communities. These communities will have a very good conductance, but they do not correspond well to the ground truth. In some sense the optimizer is `too good', and conductance is not the best criterion to describe these communities.
A better objective would perhaps take into account the size of the community more explicitly, but this needs to be investigated further.

In this paper we have used gradient descent, a first order optimization method which utilizes only the objective function's gradient. More advanced optimization methods also use second derivatives or approximations of those.
We believe that such methods will not bring a large advantage compared to gradient descent, because during the optimization many coordinates are at the boundary value 0 or 1, and second derivatives would not help to locate these boundary points. Other constrained optimizers such as interior point methods have the problem that they need to inspect a much larger part of the network, potentially all of it, because intermediate steps have nonzero membership for all nodes.




\section*{Acknowledgments}
We thank the reviewers for their useful comments.
This work has been partially funded by the Netherlands Organization for Scientific Research (NWO) within the EW TOP Compartiment 1 project 612.001.352.

\bibliography{clustering}

\end{document}